\documentclass{IEEEtran}
\IEEEoverridecommandlockouts
\usepackage{cite}
\usepackage{amsmath,amssymb,amsfonts}
\usepackage{algorithmic}
\usepackage{graphicx}
\usepackage{amsthm}
\usepackage{textcomp}
\usepackage{xcolor}
\usepackage{booktabs}
\usepackage{tabularx} 
\usepackage{makecell}

\usepackage{tikz}
\usetikzlibrary{
    positioning,  
    arrows.meta,  
    calc,         
    decorations.pathmorphing 
}
\newtheorem{theorem}{Theorem}
\newtheorem{definition}{Definition}

\def\BibTeX{{\rm B\kern-.05em{\sc i\kern-.025em b}\kern-.08em
    T\kern-.1667em\lower.7ex\hbox{E}\kern-.125emX}}

\newcommand{\Dpc}{D_{\mathrm{pc}}}

\newcommand{\Ihat}{\hat{I}}

\begin{document}

\title{A Copula-based Semantics-Structure Minimization Framework for QoS Guaranteed

Wireless Communications}

\author{Xinke~Jian,
        Zhiyuan~Ren,~\IEEEmembership{Member,~IEEE},
        Wenchi~Cheng,~\IEEEmembership{Senior~Member,~IEEE}
\thanks{This work was supported by the National Key Research and Development Program of China (No.2023YFC3011502).}%
\thanks{Xinke~Jian, Zhiyuan~Ren and Wenchi~Cheng are with the School of Telecommunications Engineering, Xidian University, Xi'an 710071, China.}
\thanks{Corresponding author: Zhiyuan Ren (zyren@xidian.edu.cn).}%
}

\maketitle
\thispagestyle{empty}

\begin{abstract}
Current empirically driven research on semantic communication lacks a unified theoretical foundation, preventing quantifiable Quality of Service (QoS) guarantees, particularly for transmitting minimal structural semantics in emergency scenarios. This deficiency limits the field's evolution into a predictable engineering science.
To address this, we establish a complete theoretical axiomatic basis for this problem.
We propose four axioms and rigorously prove that the family of pairwise rank-Copulas, $\{C_\delta\}$, is the minimal sufficient representation for \textbf{minimal} structural semantics.
Based on this, we construct a semantic distortion metric, $\Dpc$, centered on the Jensen–Shannon (JS) divergence.
We then establish the core theoretical boundaries of the framework: (1) sample complexity bounds; (2) rate-distortion bounds;
(3) an end-to-end SLA reachability theorem; and (4) a semantic source-channel separation theorem, which provides a provable QoS guarantee.
Finally, we validate our framework through decoupled experiments, empirically demonstrating that our core metric ($\Dpc$) strictly adheres to our foundational axioms while standard perceptual metrics (like LPIPS) fail to do so.
\end{abstract}


\section{Introduction}

Since Shannon's foundational work, the guiding principle of information theory has been the exact physical replication of data \cite{Shannon1948}.
This ``bit-perfect" paradigm, where the gold standard is an infinitesimally low Bit Error Rate (BER), has undeniably built our modern digital world.
However, with network traffic now dominated by high-dimensional data such as images, videos, and Extended Reality (XR) streams, this paradigm is reaching a fundamental efficiency limit \cite{Gunduz2023}.
The practice of expending vast channel resources to protect every single bit, many of which are inconsequential to the underlying``meaning", is becoming an unsustainable bottleneck, especially in resource-constrained emergency communications\cite{wang2020statistical}.

This challenge has catalyzed the rise of \textbf{Semantic Communication}, 
a transformative paradigm poised to be the cornerstone of future 6G networks 
\cite{lu2024semcom,yang2023semcom}.
Instead of replicating bits, its goal is to efficiently transmit the core meaning, or semantics. In an emergency context, this translates to conveying minimal mission-critical information that is relevant to the end user, be it a human or an AI agent.
Driven by the success of deep learning, the field has seen a surge of empirical research, 
mainly centered on end-to-end trained neural network architectures 
\cite{xie2021jsac_image,xie2021tsp_sc}. 
Although these pioneering works have demonstrated impressive compression gains, 
their reliance on empiricism leaves the field on unstable ground. 
This ``black-box'' approach creates three fundamental gaps that prevent semantic 
communication from evolving into predictable engineering science.

First, there is a \textbf{definitional gap}. ``semantics'' itself lacks a rigorous, 
universally accepted mathematical definition. 
Second, this leads to a \textbf{measurement gap}. 
In the absence of a true semantic fidelity metric, researchers are forced to repurpose 
measures of perceptual similarity, such as LPIPS and DISTS 
\cite{zhang2018lpips,ding2022dists}. 
However, as we argue in our work, perceptual fidelity is not synonymous with 
semantic fidelity. 
This misalignment means that optimizing for perception does not guaranty the 
preservation of task-critical semantic information\cite{yan2022qoe}. 
Third, without a solid definition and a reliable metric, a \textbf{theoretical gap} 
is inevitable. 
It is impossible to establish performance limits or provide deterministic 
service guarantees, such as Service Level Agreements (SLAs) 
\cite{behrouz2019sla}, and thus impossible to provide users with quantifiable semantic QoS guarantees.

To bridge these gaps, this paper returns to the first principles, building semantics 
from a solid axiomatic foundation instead of learning it from data. Our core thesis is that the essence of \textbf{minimal structural semantics} lies not 
in pixel values, but in their intrinsic dependency structure, which is 
invariant to monotonic transformations (e.g. brightness changes).
This focus is 
by design, allowing us to isolate core content from nuisance variables like illumination.
We propose four axioms (monotone invariance, pairwise stationarity, minimal 
sufficiency, and metric stability) and rigorously prove that the minimal sufficient 
representation for these minimal structural semantics is the family of its 
\textbf{pairwise rank-Copulas}, $\{C_\delta\}$.

This constructive approach elevates Copula theory from a heuristic tool \cite{nelsen2006copula,sklar1959} 
to a foundational element of communication.
Unlike prior heuristic use, our contributions are: 
1) an axiomatic foundation rigorously proving that 
copulas are the minimal sufficient representation; 
2) a derived distortion metric $D_{pc}$; 
3) a set of theoretical bounds including an end-to-end 
SLA theorem; and 4) experimental validation of the framework.

\section{Axiomatic Foundation and Semantic Representation}
\label{sec:axioms_and_rep}

\subsection{The Axiomatic System}
We postulate that the image signal $I$ can be modeled as a Translation-Invariant Pairwise Markov Random Field (MRF). This is a strong assumpation (as discussed in Sec. \ref{sec:scope}). However, it is crucial for achieving mathematical tractability and provides a principled starting point, as visual primitives like edges and textures are defined by local pairwise dependencies. Under this model class, we define an equivalence class for minimal structural semantics based on four axioms.

\begin{itemize}
    \item \textbf{Axiom A1 (Monotone Invariance):} The structural semantics must be strictly invariant to any strictly monotonic point-wise transformation $T$ (e.g. gamma correction, contrast adjustment). This axiom mathematically separates the semantic structure from its attributes (e.g. brightness).
    
    \item \textbf{Axiom A2 (Pairwise Stationarity):} The statistical properties of the image are sufficiently described by pairwise dependencies that are stationary (translation-invariant) over a finite set of spatial displacement vectors $\Delta = \{\delta\}$.
    The choice of $\Delta$ is critical as it defines the scope of the structural semantics being captured. It is assumed to cover all dependencies of interest (i.e. all edges in the assumed MRF). An insufficient $\Delta$ would lead to semantic approximation error, as noted in Sec. \ref{sec:scope}.
    
    \item \textbf{Axiom A3 (Minimal Sufficiency):} The semantic representation must be statistically complete (sufficient) and contain no redundant information (minimal) for the model class defined by A1 and A2.
    
    \item \textbf{Axiom A4 (Decomposable \& Stable Distance):} A valid semantic distortion metric $D(I, \Ihat)$ must (i) be decomposable over independent semantic components (defined by $\delta$) and (ii) satisfy the Data Processing Inequality, ensuring stability under coarse-graining.
\end{itemize}

\subsection{Derivation of Representation and Metric}
From this axiomatic system, we derive the representation and metric as logical necessities, as illustrated in Fig.~\ref{fig:copula_process}. First, \textbf{Axiom A1 (Monotone Invariance)} compels the use of the Rank Transform, which combined with Sklar's Theorem \cite{sklar1959}, isolates the \textbf{Copula} as the sole carrier of the pure dependency structure (see Fig. \ref{fig:copula_example} for a concrete example).

Next, \textbf{Axiom A2 (Pairwise Stationarity)} enables pairwise sampling, and \textbf{Axiom A3 (Minimal Sufficiency)} leads directly to our first key theoretical result. We first formally define the model class under consideration: let $\mathcal{M}_{\Delta}$ be the class of all image distributions that are consistent with Axioms A1-A2, and whose dependency structure is fully captured by the pairwise relationships defined over the displacement set $\Delta$. Within this class, we have:

\begin{figure}[t]
\centering
\includegraphics[width=0.5\textwidth]{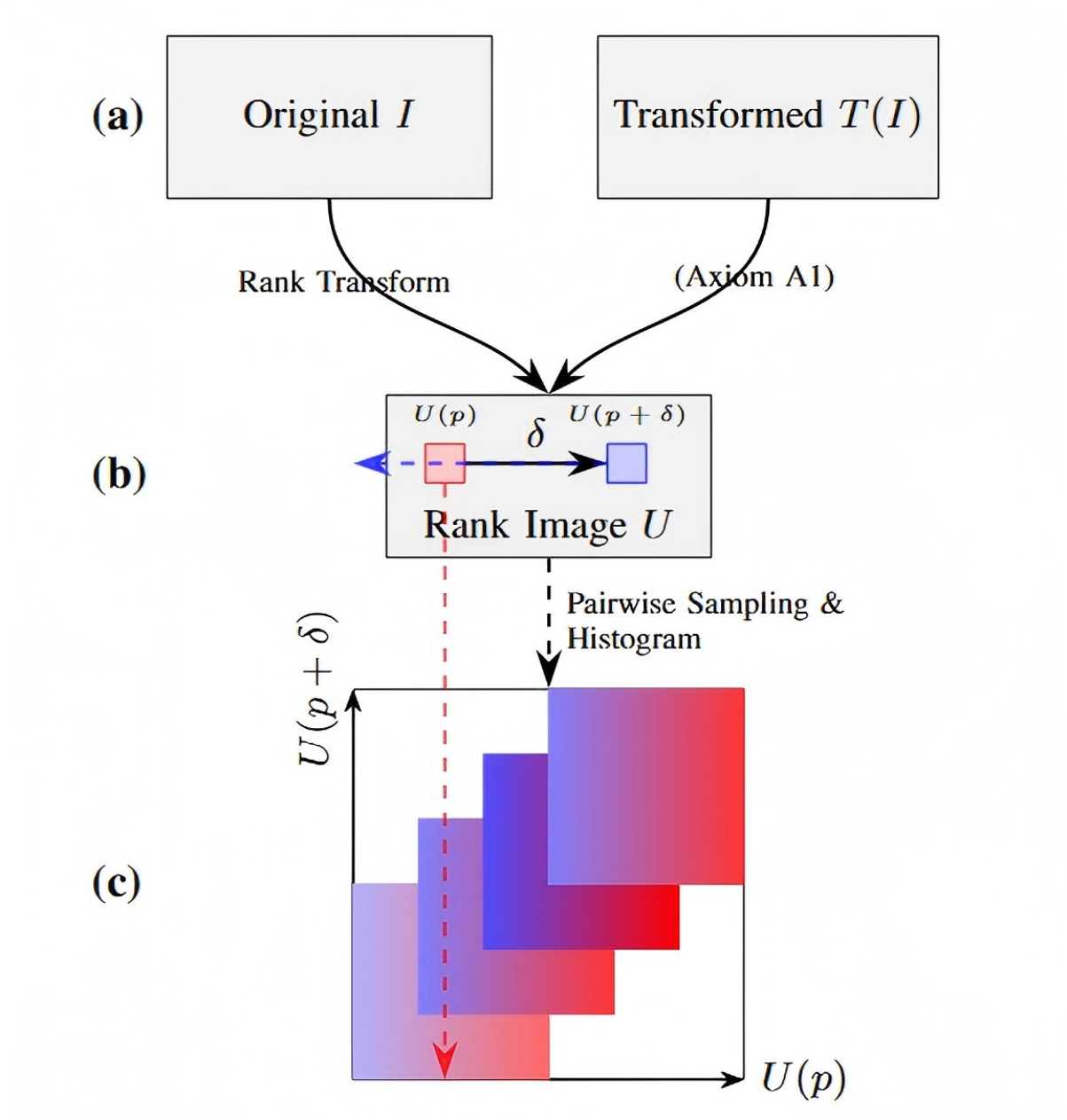}
\caption{Illustration of the derivation from an image to its Copula representation. (a) An original image $I$ and its monotonically transformed version $T(I)$ have different perceptual appearances but identical structural semantics. (b) The Rank Transform, enforced by Axiom A1 , maps both images to the \textbf{same} unique rank image $U$. (c) For a given displacement $\delta$, pairwise sampling extracts the rank relationships $\big(U(p), U(p+\delta)\big)$ . These pairs are collected into a 2D histogram, which forms the empirical Copula $C_\delta$—our "relationship fingerprint".}
\label{fig:copula_process}
\vspace{-0.2cm}
\end{figure}

\begin{figure}[htbp]
    \centering
    \includegraphics[width=0.5\textwidth]{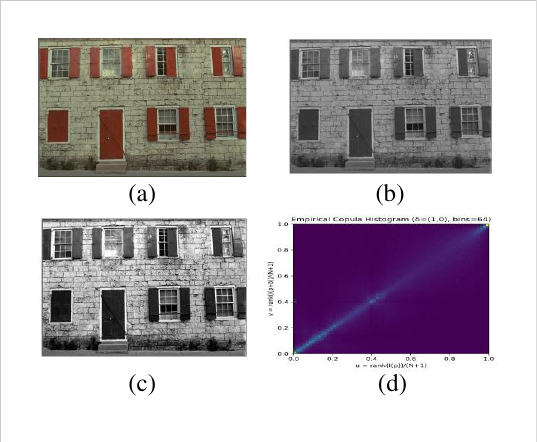}
    \vspace{-0.5cm}
    \caption{Validation of Axiom A1 (Monotone Invariance). (a) Original $I$, (b) grayscale $T_1(I)$, and (c) high-contrast $T_2(I)$ are 
semantically equivalent under A1 despite perceptual differences. (d) All three produce the identical empirical Copula $C_{\delta}$ (for $\delta=(1,0)$), 
proving the Copula isolates the dependency structure from marginal statistics}
    \label{fig:copula_example}
    \vspace{-0.4cm}
\end{figure}

\begin{theorem}[Minimal Sufficient Representation of Minimal Structural Semantics]
In the model class defined by Axioms A1 (Monotone Invariance) and A2 (Pairwise Stationarity), the family of pairwise rank-Copulas $\{C_\delta(I)\}_{\delta \in \Delta}$ is the minimal sufficient statistic for the dependency structure of the image $I$.
\end{theorem}

\begin{definition}[Pairwise Rank-Copula Semantic Representation]
The minimal sufficient structural semantic representation of an image $I$ is the set of its empirical pairwise Copulas, indexed by the displacement set $\Delta$:
\begin{equation}
\label{eq:copula_rep}
\mathcal{C}(I) = \{C_\delta(I) : \delta \in \Delta\}
\end{equation}
where $C_\delta(I)$ is the 2D Copula estimated from all valid pairs of rank-normalized pixel values $(U(p), U(p+\delta))$.
\end{definition}

Finally, \textbf{Axiom A4 (Decomposable \& Stable Distance)} guides the metric selection, requiring properties (e.g., Data Processing Inequality) met by f-divergences. Symmetry and the true metric property (when square-rooted) point towards the Jensen-Shannon (JS) divergence as an ideal building block.

\begin{definition}[Pairwise Copula Consistency ($\Dpc$)]
The structural semantic distortion between an original image $I$ and a reconstructed image $\Ihat$ is defined as the average Jensen-Shannon distance (the square-root of the JS divergence) 
between their Copula representations:
\begin{equation}
\label{eq:dpc_metric}
D_{pc}(I,\hat{I}) := \frac{1}{|\Delta|}\sum_{\delta\in\Delta}\sqrt{JS(C_{\delta}(I),C_{\delta}(\hat{I}))}
\end{equation}
where $\mathrm{JS}(P, Q)$ is the Jensen-Shannon divergence.
\end{definition}

This metric inherits the stability required by \textbf{Axiom A4}. 
Since the JS-divergence satisfies the Data Processing Inequality (DPI) for each $\delta$, 
and both the square-root function ($\sqrt{\cdot}$) (being monotonic) 
and the averaging operator ($\frac{1}{|\Delta|}\sum$) (being linear) 
preserve this inequality, the overall $D_{pc}$ metric also satisfies the DPI.

\section{Theoretical boundary Guarantees}
\label{sec:theory}

Having defined the representation $\mathcal{C}(I)$ (Eq. \ref{eq:copula_rep}) and the metric $\Dpc$ (Eq. \ref{eq:dpc_metric}), we can now establish the fundamental theoretical limits of this framework.
We derive bounds for both average distortion (e.g. $\mathbb{E}[D_{pc}]$ in Rate-Distortion theory) and high-probability guarantees (e.g. $\Pr[D_{pc} \le \varepsilon]$ in our SLA theorem).
As $\Dpc$ is a bounded metric, these two types of bounds are formally linked via concentration inequalities (e.g. Markov's inequality), though we will derive them separately for each theorem's specific context.

\subsection{Sample Complexity of Semantic Representation}
First, we establish the sample complexity required for a statistically reliable estimate of the Copula $\hat{C}_\delta$.

\begin{theorem}[Empirical Copula Concentration Bound]
\label{thm:concentration}
Let the empirical Copula $\hat{C}_\delta$ be estimated as a 2D histogram on a $B \times B$ grid, using $n_{\mathrm{eff}}$ effective i.i.d. sample pairs (e.g., obtained via non-overlapping sub-sampling of the pixel pairs). To guarantee an average $L_1$ estimation error $\frac{1}{|\Delta|}\sum_{\delta} \|C_\delta - \hat{C}_\delta\|_1 \le t$ with a confidence of at least $1-\eta$, the number of samples $n_{\mathrm{eff}}$ must satisfy:
\begin{equation}
n_{\mathrm{eff}} \ge \frac{2}{t^2}\ln(\frac{2^{B^2}|\Delta|}{\eta})
\end{equation}
\end{theorem}

\textit{Significance:} This theorem quantifies the ``sample cost" of structural semantics,
revealing that the required number of samples $n_{\mathrm{eff}}$ grows exponentially with the quantization granularity $B^2$,
making high-fidelity estimation data-intensive.

\subsection{Rate-Distortion Bounds for Semantic Compression}
Second, we establish the rate-distortion (R-D) function for compressing this semantic representation\cite{berger1971rate}.

\begin{theorem}[Semantic Rate-Distortion Bounds]
\label{thm:rd_bounds}
Let $R_{pc}(D)$ be the rate-distortion function for the semantic representation $\mathcal{C}(I)$ under the $\Dpc$ metric. For a representation with $|\Delta|$ displacement vectors, each estimated on a $B \times B$ grid (with $B^2-1$ degrees of freedom):
\begin{enumerate}
    \item \textbf{(Achievability)} There exists a quantization and entropy coding scheme with rate $R$ and semantic encoding distortion $\varepsilon_{\mathrm{enc}}$ satisfying:
    \begin{equation}
R\le|\Delta|(B^{2}-1)log_{2}\frac{1}{\alpha} \quad D_{pc,enc} \le C' B^{2}\alpha
\end{equation}

    where $\alpha$ is the quantization step size.
    \item \textbf{(Converse)} Any encoding scheme that achieves an average semantic distortion $\varepsilon_{\mathrm{enc}}$ must have a rate $R$ that is lower-bounded by:
    \begin{equation}
    \label{eq:rd_converse}
    R \ge c \cdot |\Delta|(B^2-1)\log_2\frac{1}{\varepsilon_{\mathrm{enc}}} - \mathcal{O}(1)
    \end{equation}
\end{enumerate}
\end{theorem}
\textit{Significance:} This theorem defines the ``compressibility" of structural semantics.
It provides a converse bound as a benchmark and an achievability bound, $g(R) \propto 2^{-R/d}$ (where $d=|\Delta|(B^2-1)$),
which is a key input for our SLA theory.

\subsection{End-to-End SLA Reachability}
While Theorems \ref{thm:concentration} and \ref{thm:rd_bounds} govern individual stages, a practical system must manage the end-to-end accumulation of errors. Theorem \ref{thm:sla} provides this guarantee.

\begin{theorem}[$(1-\delta_{\mathrm{SLA}})$-SLA Reachability]
\label{thm:sla}
Consider a complete semantic communication system where:
\begin{itemize}
    \item Copula estimation (Thm. \ref{thm:concentration}) introduces an error $\varepsilon_{\mathrm{est}}$ with confidence $1-\eta_{\mathrm{est}}$.
    \item Semantic encoding (Thm. \ref{thm:rd_bounds}) introduces a quantization distortion $\varepsilon_{\mathrm{enc}}$.
    \item The generative decoder introduces a decoding/rendering error $\varepsilon_{\mathrm{dec}}$ with confidence $1-\eta_{\mathrm{dec}}$.
\end{itemize}
Then, the total end-to-end structural semantic distortion $\Dpc(I, \Ihat)$ is probabilistically bounded by:
\begin{equation}
Pr[D_{pc}(I,\hat{I}) \le D_{pc,est} + D_{pc,enc} + D_{pc,dec}] \ge 1-\delta_{SLA} \quad
\end{equation}
where $\delta_{\mathrm{SLA}} \le \eta_{\mathrm{est}} + \eta_{\mathrm{dec}}$ (by the union bound) is the overall system outage probability.
\end{theorem}

\textit{Significance:} This theorem is the cornerstone of predictable semantic communication,
providing a ``map" from component performance ($\varepsilon_{\mathrm{est}}, \varepsilon_{\mathrm{enc}}, \varepsilon_{\mathrm{dec}}$)
to end-to-end system reliability ($\varepsilon, \delta_{\mathrm{SLA}}$), 
which enables the R-C-R trade-offs analyzed in Sec. \ref{sec:evaluation}.

\subsection{Semantic Source-Channel Separation}
Finally, we question if Shannon's classic separation theorem still holds in our semantic framework. Theorem \ref{thm:separation} confirms that it does.

\begin{theorem}[Semantic Source-Channel Separation]
\label{thm:separation}
Let the image source $I$ be a stationary and ergodic process, and let the channel be a discrete memoryless channel (DMC) with capacity $C$.
Define the structural semantic rate-distortion function as $R_{pc}(D) = \inf I(I;\hat{I}) \text{ s.t.
} \mathbb{E}[\Dpc(I,\hat{I})] \le D$.
Given that our $D_{pc}$ metric (Eq. 2) is bounded and continuous,
if $C > R_{pc}(\varepsilon)$, there exists a separable source and channel coding scheme that can achieve an average end-to-end structural semantic distortion $\mathbb{E}[\Dpc(I, \hat{I})]$ arbitrarily close to $\varepsilon$.
\end{theorem}

\textit{Significance:} This theorem justifies a modular ``source coding" + ``channel coding" design,
proving that separating them incurs no optimality loss in this framework.
\section{Performance Evaluation}
\label{sec:evaluation}

Given the significant engineering challenges in building a fully optimized end-to-end decoder, we adopt a decoupled experimental methodology. Our goal is not to evaluate a specific prototype, but to validate the core tenets of our theoretical framework. We conduct a two-part validation:
\begin{enumerate}
    \item First, we validate the practical effectiveness of our metric, $\Dpc$, as a task-oriented ``ruler" for AI agent performance.
    \item Second, we validate the theoretical soundness of this ruler by verifying its strict compliance with our foundational axioms.
\end{enumerate}

\subsection{Experimental Setup}
\begin{itemize}
    \item \textbf{Datasets:} We use the \textbf{Kodak PhotoCD dataset}  for axiom validation and the \textbf{COCO val2017} dataset for task-oriented validation.
    \item \textbf{Task-Oriented Evaluation:} We use a pre-trained \textbf{YOLOv8n} model to evaluate object detection mAP (mAP50-95) on distorted COCO images. We test 15 degradation levels across three distortion types: \textbf{JPEG}, \textbf{Gaussian Blur}, and \textbf{Gaussian Noise}.
    \item \textbf{Baseline Metrics:} We compare our $\Dpc$ against the state-of-the-art perceptual metric \textbf{LPIPS} and the classic pixel-fidelity metric \textbf{PSNR}.
    \item \textbf{Core Metric:} We compute the \textbf{Spearman Rank Correlation Coefficient (SRCC, $\rho$)} between the degradation in mAP and the distortion reported by each metric.
\end{itemize}

\subsection{Validation of $\Dpc$ as a Task-Oriented Metric (The ``Ruler")}
The goal of semantic communication is to efficiently transmit the core meaning relevant to an end user, be it a human or an AI agent. Therefore, a valid semantic distortion metric must be a strong predictor of downstream AI task performance. Intuitively, we hypothesize that our $\Dpc$ metric (which measures low-level structure) should correlate with high-level AI tasks (like mAP). We believe this correlation exists because the Rank-Copula representation naturally captures the essential visual primitives—such as edges, textures, and topological dependencies—that deep neural networks rely on for feature extraction and object recognition. We tested this hypothesis directly by measuring the mAP of the YOLOv8n object detector on all 15 distorted image sets.

The resulting mAP was plotted against the distortion values reported by $\Dpc$ (Semantic), LPIPS (Perceptual), and PSNR (Pixel). The results are presented in Fig. \ref{fig:map_vs_metrics}.
\begin{itemize}
    \item \textbf{Pixel vs. Classic Structure:} The results clearly show that pixel-level fidelity (PSNR) and classic structural similarity (SSIM) are poor predictors of AI task performance, exhibiting the weakest correlations ($\rho = 0.750$ for PSNR, and $\rho = 0.782$ for SSIM).
    \item \textbf{High Correlation:} In stark contrast, both the semantic metric $\Dpc$ ($\rho = -0.896$) and the perceptual metric LPIPS ($\rho = -0.961$) demonstrate an extremely strong negative correlation with mAP.
\end{itemize}

This experiment confirms our framework's core premise: communication systems designed for AI agents must be optimized for structural or perceptual fidelity, not pixel fidelity. This result (LPIPS $>$ $\Dpc$ in correlation) highlights a key finding. 
While LPIPS shows high correlation (Fig. \ref{fig:map_vs_metrics}), 
it is theoretically unreliable as it violates Axiom A1 
(Table \ref{tab:axiom_results}), confusing attributes with structure. 
Therefore, $\Dpc$ is the only metric evaluated that is both 
practically effective (strong task correlation) and theoretically 
sound (axiom-compliant).

\begin{figure*}[t]
\centering
\includegraphics[width=0.9\textwidth]{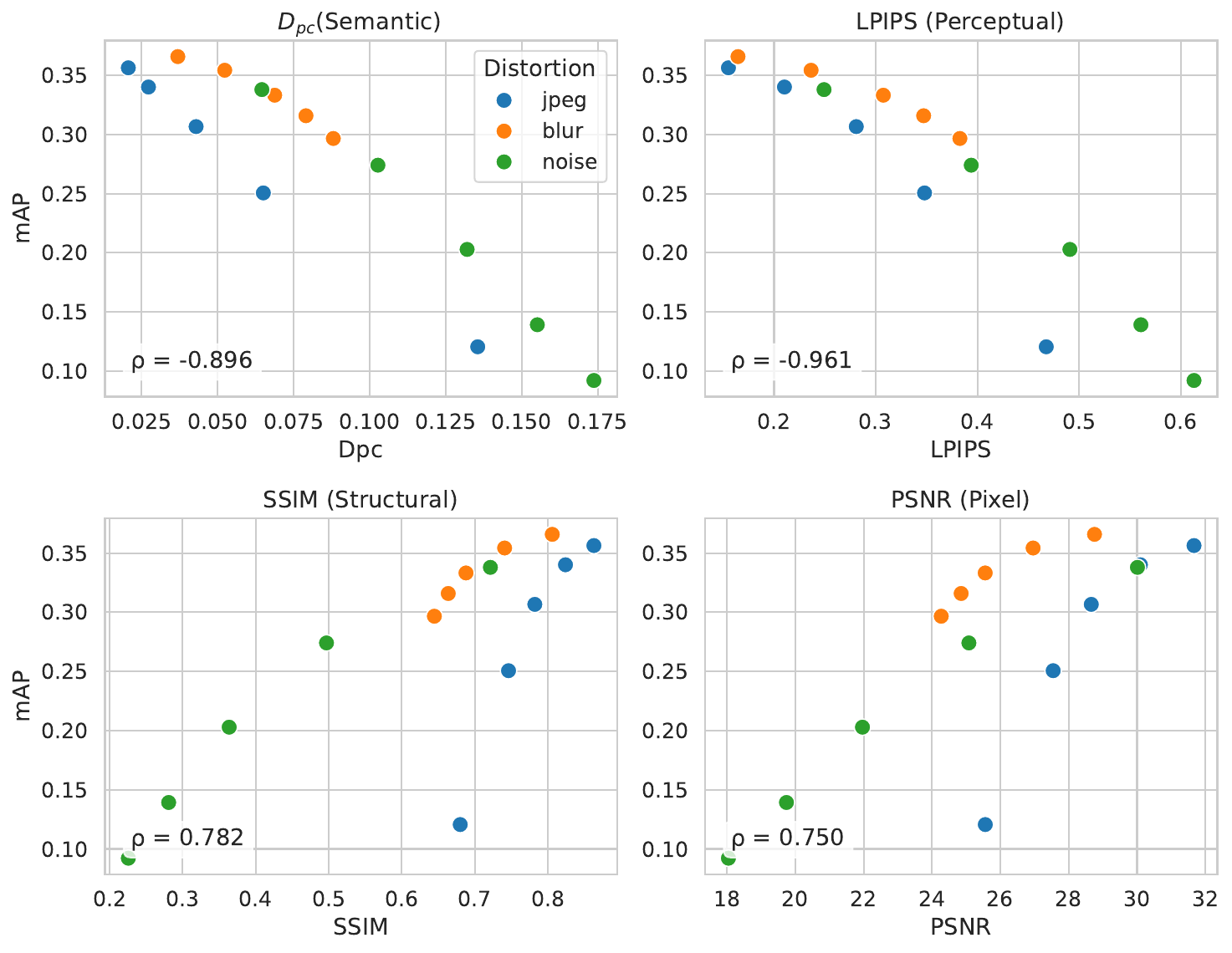}
\vspace{-0.3cm}
\caption{Correlation between AI task performance (mAP) and distortion metrics across 15 degradation levels (JPEG, Blur, Noise). The SRCC ($\rho$) values show that both $D_{\mathrm{pc}}$ (Semantic) and LPIPS (Perceptual) are strong predictors of mAP degradation, while PSNR (Pixel) is significantly less correlated.}
\label{fig:map_vs_metrics}
\vspace{-0.4cm}
\end{figure*}

\subsection{Validation of Foundational Axioms (The ``Soundness")}
\label{sec:axiom_validation}
A metric is only useful if it reliably adheres to the foundational axioms that define it. In Sec. IV-B, we established that LPIPS and $\Dpc$ are both effective task predictors. We now test if they are also theoretically sound by validating them against \textbf{Axiom A1 (Monotone Invariance)}. This axiom states that structural semantics must be strictly invariant to monotonic transformations like brightness or contrast changes.

We applied two classes of transformations to a source image: (1) Semantic-Invariant (monotonic) transformations and (2) Structure-Damaging transformations. The results are shown in Table \ref{tab:axiom_results}.

\begin{table}[t] 
\caption{Validation of Axiom A1 (Monotone Invariance)}
\label{tab:axiom_results}
\centering 
\small 
\begin{tabularx}{\columnwidth}{l c c c c} 
\toprule
\textbf{Transform} & \makecell{\textbf{Semantic}\\\textbf{Change?}} & \textbf{$\Dpc$ (Ours)} & \textbf{LPIPS} & \textbf{SSIM} \\
\midrule
\multicolumn{5}{l}{\textit{Semantic-Invariant (Axiom A1)}} \\
Brightness +80 & NO & \textbf{0.00000} & 0.0686 & 0.8447 \\
Gamma 0.5      & NO & \textbf{0.00421} & 0.0602 & 0.7157 \\
\midrule
\multicolumn{5}{l}{\textit{Structure-Damaging}} \\
Blur (7, 7)   & YES & 0.05197 & 0.3366 & 0.6249 \\
JPEG (Q=20)   & YES & 0.01544 & 0.2866 & 0.8206 \\
\bottomrule
\end{tabularx}
\vspace{-0.3cm}
\end{table}

The results are clear. Both the perceptual metric (LPIPS) and the classic structural metric (SSIM) do not satisfy the requirements of this test.
They incorrectly report significant distortion (e.g. LPIPS: 0.0686, SSIM: 0.8447) for semantic-invariant brightness and gamma changes. 
This demonstrates that both LPIPS and SSIM are unreliable rulers, as they fundamentally confuse perceptual attributes with structural semantics.

In stark contrast, our $\Dpc$ metric is in strong agreement with our theory.
It correctly reports a distortion of \textbf{0.00000} for the pure monotonic brightness shift and only a near-zero quantization error for the gamma correction.
This two-part evaluation confirms that $\Dpc$ is the only metric evaluated that is both practically effective (Fig. \ref{fig:map_vs_metrics}) and theoretically sound (Table \ref{tab:axiom_results}).

\subsection{Validation of Theoretical Bounds (The ``Map")}
We then conducted experiments to verify the accuracy of the theoretical "map" derived in Sec. \ref{sec:theory}.

\begin{figure}[htbp]
\centerline{\includegraphics[width=1.0\columnwidth]{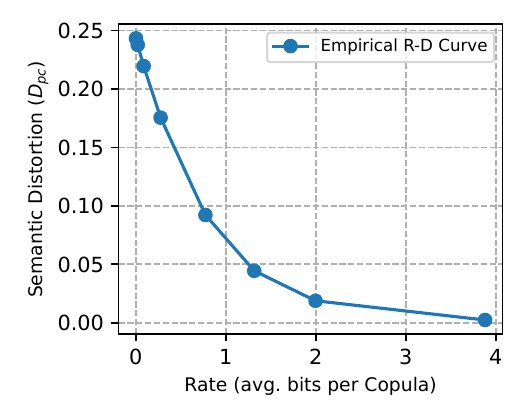}}
\caption{Empirical validation of the rate-distortion (R-D) performance of the Copula semantic representation (data from Sec.~\ref{sec:evaluation}.C.1). This curve confirms the compressibility of semantics and provides the $g(R)$ function for our SLA theory (Thm. \ref{thm:rd_bounds}).}
\label{fig:rd_curve}
\end{figure}
\paragraph{Rate-Distortion of Semantic Representation (Validating Thm. 3)}

First, we validated the semantic R-D bounds.
We extracted the empirical Copula $\{C_{\delta}\}$, applied uniform scalar quantization (step $\alpha$), 
and calculated the entropy (Rate) vs. semantic distortion ($D_{pc}$) to generate the R-D curve.
The resulting curve (Fig. \ref{fig:rd_curve}) empirically confirms the compressibility 
of semantics and provides the $g(R)$ function for our SLA model.

The primary value of our framework is that it moves system design from empiricism to science. Our architecture is directly guided by the end-to-end SLA guarantee established in our theoretical framework. As proven in \textbf{Theorem \ref{thm:sla}}, the overall semantic performance is predictable. 
This theoretical result is not just an academic statement; it yields a concrete engineering design tool that quantifies the fundamental trade-off between three competing resources: transmission \textbf{Rate ($R$)}, decoder \textbf{Computation ($T$)}, and semantic \textbf{Reliability ($\varepsilon$)}.
For any target SLA $(\varepsilon, \delta_{\mathrm{SLA}})$, the relationship between these resources can be expressed.
From our linear SLA bound in Thm. \ref{thm:sla}, we have $\varepsilon \ge D_{pc,est} + g(R) + f(T)$.
This leads to the minimum required rate $R_{\min}$ as a function of the allotted computation $T$:
\begin{equation}
R_{tmin}(T)\approx C_{1}\cdot \log_{2} \left(\frac{C_{2}}{\epsilon-\epsilon_{est}-f(T)}\right)
\end{equation}
Conversely, the minimum computation $T_{\min}$ is a function of the available rate $R$:
\begin{equation}
T_{min}(R)\approx C_{3}\cdot \log \left(\frac{C_{4}}{\epsilon-\epsilon_{est}-g(R)}\right)
\end{equation}
where $f(T)$ models the error reduction from computation, and $g(R)$ models the error reduction from rate (with $g(R) \propto 2^{-R/d}$, 
where $d=|\Delta|(B^2-1)$, due to the $O(\alpha)$ bound in Thm. 3
under the $D_{pc}$ metric).
These equations form the ``map" for designing predictable semantic systems, allowing engineers to trade bandwidth for computation to meet a specific semantic guarantee.

\paragraph{Rate-Computation-Reliability Trade-off (Validating Thm. 4)}
Second, to illustrate the principle of the SLA framework, we validated its 
trade-offs. We simulated the 3D trade-off surface by combining our empirical 
R-D curve (for $g(R)$) with a hypothesized
decoder convergence model (for $f(T)$). The resulting surface (Fig. 
\ref{fig:sla_surface}) demonstrates how Thm. \ref{thm:sla} serves as a 
practical design ``map", allowing engineers to quantitatively trade rate 
for computation to meet a target reliability ($\epsilon$).

\vspace{-0.2cm}
\begin{figure}[htbp]
\centerline{\includegraphics[width=1.0\columnwidth]{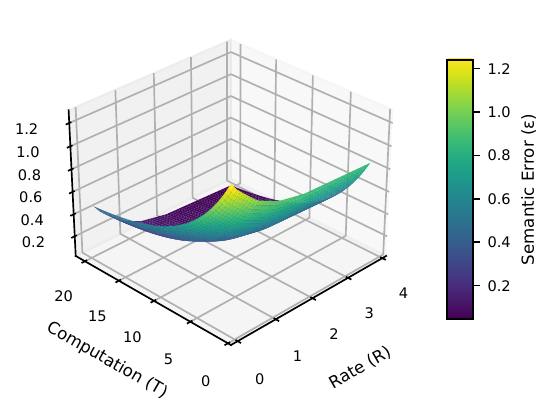}}
\vspace{-0.3cm}
\caption{Visualization of the Rate-Computation-Reliability trade-off surface derived from our SLA theorem (Thm. \ref{thm:sla})  and empirical data from Fig. \ref{fig:rd_curve}. This "map" allows engineers to quantitatively trade bandwidth (Rate) for computation (Computation) to meet a specific semantic reliability (Error) target.}
\label{fig:sla_surface}
\vspace{-0.5cm}
\end{figure}

\section{DISCUSSION AND CONCLUSION}
\label{sec:scope}
This paper proposed an axiomatic foundation for structural semantic communication,
yielding a reliable metric ($D_{pc}$) and a predictive SLA framework (Thm. 4) that 
quantifies the Rate-Computation-Reliability trade-off. It is critical, however, to define the precise scope of our framework. 
By design (Axiom A1), our theory isolates structural semantics, intentionally 
filtering out ``attribute semantics" (e.g. color or absolute values in 
medical imaging). These are carried by the marginal distributions, 
and a practical system requiring them would need to supplement our structural 
representation with a separate channel.

A second major limitation is the ``decoder gap". The end-to-end 
realization of our SLA guarantee (Thm. \ref{thm:sla}) relies on a generative 
decoder capable of solving the challenging inverse problem of reconstructing an 
image from its Copula representation $\{C_{\delta}\}$. While implementing 
such a decoder is a critical future challenge, a key contribution of our 
framework is providing a clear and theoretically-grounded optimization target 
for this task: using $\Dpc$ as a novel loss function for generative models 
(e.g. Diffusion \cite{ho2020denoising} or GANs \cite{goodfellow2014generative}).

Finally, our reliance on Axiom A2 (Pairwise Stationarity) is a 
simplification for mathematical tractability; extending this framework to 
capture higher-order dependencies is a significant open research question. 
Despite these limitations, the framework is highly general and can be extended to other modalities, such as video (3D spatio-temporal fields) and 
audio (2D spectrograms) . 

This work provides rigorous definitions, a validated metric, and predictable bounds
to help transition semantic communication from an empirical art to a designable 
engineering science.

\appendix
\newtheorem{lemma}{Lemma}
\newtheorem{corollary}{Corollary}
\section{Proofs of Theorems}
\label{sec:appendix_proofs}

\subsection{Proof of Theorem 1 (Minimal Sufficient Representation)}

\begin{theorem}[Minimal Sufficient Representation of Structural Semantics]
In the model class $\mathcal{M}_{\Delta}$ defined by Axioms A1 (Monotone Invariance) and A2 (Pairwise Stationarity), where the dependency structure is assumed to be fully captured by the displacement set $\Delta$, the family of pairwise rank-Copulas $\{C_\delta(I)\}_{\delta \in \Delta}$ is the minimal sufficient statistic for the dependency structure of the image $I$.
\end{theorem}

\begin{proof}
Let the pixel random field after the rank transform be $U=\{U(p)\}_{p\in\Omega}$. Its marginal distributions are all $\mathrm{Unif}(0,1)$ (as required by Axiom A1).

\begin{enumerate}
    \item \textbf{Sufficiency}: By Sklar's Theorem, when the marginal distributions are fixed (to uniform), any finite-dimensional joint distribution is uniquely determined by its Copula function. Under the Pairwise Markov Random Field (MRF) assumption (Axiom A2), the log-likelihood function of the image can be decomposed according to its cliques:
    $$
    \log p_\theta(U)=\sum_{\delta\in\Delta}\ \sum_{p\in\Omega_\delta}\log \phi_\delta\!\big(U(p),U(p+\delta)\big)-\Psi(\theta)
    $$
    Here, $\phi_\delta$ is a potential function that depends only on the pairwise Copula density $c_\delta$, and $\Psi(\theta)$ is the partition function. Under this decomposition, all samples $(U(p),U(p+\delta))$ influence the likelihood function only through their empirical Copula $\hat C_\delta$. Therefore, by the \textbf{Fisher–Neyman Factorization Theorem}, the family of statistics $\{\hat C_\delta\}_{\delta\in\Delta}$ is sufficient for the model parameter family $\{\phi_\delta\}$.

    \item \textbf{Minimality}: We apply the criterion for minimal sufficiency from the \textbf{Lehmann–Scheffé Theorem}. A statistic is minimally sufficient if, for any two parameter sets $\theta, \theta'$, the likelihood ratio $p_\theta(U) / p_{\theta'}(U)$ is a function of the statistic $\{\hat C_\delta(U)\}$ only. In our model, if two parameter sets correspond to the same Copula family $\{C_\delta\}$, then all potential functions $\phi_\delta$ are identical, leading to a constant likelihood ratio (independent of the specific sample $U$). This satisfies the condition for minimal sufficiency.
\end{enumerate}
Thus, the Copula family $\{C_\delta(I)\}_{\delta\in\Delta}$ is the minimal sufficient statistic for this model class.
\end{proof}

\subsection{Proof of Theorem 2 (Empirical Copula Concentration)}

\begin{theorem}[Empirical Copula Concentration]\label{thm:emp-copula-conc}
Let $\widehat C_\delta$ be the empirical copula on a $B\times B$ grid
estimated from $n_{\mathrm{eff}}$ effective i.i.d.\ pairs (e.g., via non-overlapping sub-sampling) $(U(p),U(p+\delta))$.
For any $t>0$, with probability at least $1-\eta$,

$$
\frac{1}{|\Delta|}\sum_{\delta\in\Delta}\big\|C_\delta-\widehat C_\delta\big\|_1 \le t
$$

whenever

$$
n_{\mathrm{eff}}\ge\frac{\mathbf{2}}{t^{2}}(B^{2}ln~2+ln|\Delta|+ln\frac{1}{\eta})=\frac{\mathbf{2}}{t^{2}}ln(\frac{2^{B^{2}}|\Delta|}{\eta}).
$$
\end{theorem}

\begin{proof}[Proof sketch]
Fix $\delta\in\Delta$ and collect the $B^2$ histogram counts as a multinomial proportion
$\widehat{\mathbf{p}}$ with truth $\mathbf{p}$ over $k=B^2$ cells.
By the Bretagnolle--Huber--Carol (BHC) inequality,
$$
\mathrm{Pr}\!\left(\|\widehat{\mathbf{p}}-\mathbf{p}\|_1 \ge t\right)\ \le\
2^{k}\exp\!\left(-\frac{n t^2}{2}\right).
$$
A union bound over $\delta\in\Delta$ yields
$
\mathrm{Pr}\!\big(\max_{\delta}\|C_\delta-\widehat C_\delta\|_1 \ge t\big)
\le |\Delta|\cdot 2^{B^2} e^{-n t^2/2}.
$
Since $\frac{1}{|\Delta|}\sum_\delta \|\cdot\|\le \max_\delta \|\cdot\|$, setting the RHS $\le \eta$ (...) and solving for $n_{\mathrm{eff}}$ (...) gives the stated bound.
This implicitly assumes the sample pairs are i.i.d., an assumption that can be met in practice by sub-sampling,
or replaced by bounds for $\alpha$-mixing fields, which is left for future elaboration.
\end{proof}

\subsection{Proof of Theorem 3 (Semantic Rate-Distortion Bounds)}

\begin{lemma}[JS-Lipschitz w.r.t.\ $L_1$]
\label{lem:js-lip}
For discrete distributions $P,Q$,
$$
\sqrt{\mathrm{JS}(P,Q)} \ \le\ \frac{\sqrt{\ln 2}}{2}\,\|P-Q\|_1$$
$$\quad\text{equivalently}\quad
\mathrm{JS}(P,Q) \ \le\ (\ln 2)\,\mathrm{TV}(P,Q)^2.
$$
\end{lemma}

\begin{theorem}[Semantic Rate-Distortion Bounds]
Let $R_{pc}(D)$ be the rate-distortion function for $\mathcal{C}(I)$ under the $\Dpc$ metric.
\begin{enumerate}
    \item \textbf{(Achievability)} There exists a quantization and entropy coding scheme with rate $R$ and encoding distortion $\varepsilon_{\mathrm{enc}}$ satisfying:
    $$
    R \lesssim |\Delta|(B^2-1)\log_2\frac{1}{\alpha},\quad D_{pc,enc}\ \le\ C' B^2 \alpha
    $$
    \item \textbf{(Converse)} Any scheme achieving an average distortion $\varepsilon_{\mathrm{enc}}$ must have a rate $R$ lower-bounded by:
    $$
    R \ge c\cdot |\Delta|(B^2\!-\!1)\log_2\tfrac1{\varepsilon_{\mathrm{enc}}}-\mathcal{O}(1)
    $$
\end{enumerate}
\end{theorem}

\begin{proof}

\textbf{Achievability (quantization error).}
\begin{enumerate}
    \item \textbf{Rate}: We use a uniform quantizer with step size $\alpha$ for each of the $B^2-1$ degrees of freedom. The bits required for entropy coding each degree of freedom is $\approx \log_2(1/\alpha)$. The total rate $R$ is the sum over all $|\Delta|$ Copulas, satisfying $R\le |\Delta|(B^2-1)\log_2(1/\alpha)+o(1)$.
    
    \item \textbf{Distortion}: Uniformly quantize each $B^2$-cell copula with step $\alpha$. This induces an $L_\infty$ error $\le \alpha/2$ per cell and hence $\|C_\delta-\widetilde C_\delta\|_1 \le B^2\alpha/2$. By Lemma~\ref{lem:js-lip},
    $JS(C_{\delta},\tilde{C}_{\delta})\le\frac{ln~2}{4}||C_{\delta}-\tilde{C}_{\delta}||_{1}^{2}\le\frac{ln~2}{4}(\frac{B^{2}\alpha}{2})^{2}=\frac{ln~2}{16}B^{4}\alpha^{2}$.

    Taking the square root to get the $d_\delta = \sqrt{JS}$ metric (used in $D_{pc}$)
and averaging over $\delta$ gives the bound:
$D_{pc,enc} = \frac{1}{|\Delta|}\sum_{\delta} \sqrt{JS(C_{\delta},\tilde{C}_{\delta})} \le \sqrt{\frac{ln~2}{16}B^{4}\alpha^{2}} = \frac{\sqrt{ln~2}}{4}B^2 \alpha$.
This provides the $O(\alpha)$ bound cited in the main paper's Theorem 3.
\end{enumerate}

\textbf{Converse (Rate Lower Bound):}
\begin{enumerate}
    \item \textbf{Problem Setup}: Any encoding scheme corresponds to a finite codebook $\mathcal{M}$ with $M=2^R$ codewords. To guarantee an average semantic distortion $D_{pc} \le \varepsilon_{\mathrm{enc}}$, this codebook must form an $\varepsilon$-net over the entire Copula space, where the covering radius $\varepsilon$ is proportional to $\varepsilon_{\mathrm{enc}}$.
    
    \item \textbf{Covering Number Bound}: We need to lower-bound the covering number $\mathcal{N}(\varepsilon;L_1)$, the minimum number of $L_1$-balls of radius $\varepsilon$ to cover the product space $\prod_{\delta\in\Delta}\Delta_{B^2-1}$. By the classic sphere-packing argument, the covering number is lower-bounded by the ratio of the total space volume to the volume of a single $\varepsilon$-ball. Let the total dimension be $d=|\Delta|(B^2-1)$. We have:
    $$
    \mathcal{N}(\varepsilon;L_1)\ \ge\ \frac{\mathrm{Vol}(\text{Domain})}{\mathrm{Vol}(L_1\text{-ball with radius }\varepsilon)}\ \ge\ C\cdot \varepsilon^{-d}
    $$
    
    \item \textbf{Conclusion}: The number of codewords $M$ must be at least the covering number, $M \ge \mathcal{N}(\varepsilon;L_1)$. Therefore, $2^R \ge C \cdot \varepsilon_{\mathrm{enc}}^{-d}$. Taking the logarithm of both sides yields the rate lower bound: $R\ge d\log_2(1/\varepsilon_{\mathrm{enc}})-\mathcal{O}(1)$.
\end{enumerate}
\end{proof}


\subsection{Proof of Theorem 4 (SLA Reachability - Robust Form)}

\begin{theorem}[SLA Reachability (robust form)]\label{thm:sla-robust}
Let $d_\delta(X,Y):=\sqrt{\mathrm{JS}(C_\delta(X),C_\delta(Y))}$ be the
Jensen--Shannon distance between copulas at displacement $\delta$,
and $D_{pc}(X,Y)=\frac{1}{|\Delta|}\sum_{\delta} d_\delta(X,Y)$ (as defined in the main paper).
Suppose the pipeline introduces per-$\delta$ errors
$a_\delta=d_\delta(I,\widehat I_{\mathrm{est}})$, 
$b_\delta=d_\delta(\widehat I_{\mathrm{est}},\widehat I_{\mathrm{enc}})$,
and $c_\delta=d_\delta(\widehat I_{\mathrm{enc}},\widehat I)$
corresponding respectively to estimation, encoding, and decoding stages. Then
$$
D_{pc}(I,\widehat I)\ \le\ D_{pc,est} + D_{pc,enc} + D_{pc,dec}
$$
where $D_{pc,est}=\frac{1}{|\Delta|}\sum_\delta a_\delta$ and
similarly for $D_{pc,enc}, D_{pc,dec}$.
\end{theorem}

\begin{proof}
The end-to-end semantic distortion compares the true copula of the source, $\mathcal{C}(I)$, with the copula of the reconstructed image, $\mathcal{C}(\hat{I})$. We introduce two intermediate semantic objects to decompose the total error: the empirical copula estimated from the source, $\widehat{\mathcal{C}}(I)$, and its quantized version, $\widehat{\mathcal{C}}_Q(I)$. The processing chain in the semantic space is thus:
$$
\mathcal{C}(I) \xrightarrow{\text{Estimation}} \widehat{\mathcal{C}}(I) \xrightarrow{\text{Encoding}} \widehat{\mathcal{C}}_Q(I) \xrightarrow{\text{Decoding}} \mathcal{C}(\hat{I})
$$
For each displacement $\delta \in \Delta$, we can apply the triangle inequality on the metric $d_\delta = \sqrt{JS}$ twice to bridge this chain:
\begin{align*}
d_\delta(\mathcal{C}(I), \mathcal{C}(\hat{I})) \le d_\delta(\mathcal{C}(I), \widehat{\mathcal{C}}(I)) + d_\delta(\widehat{\mathcal{C}}(I), \mathcal{C}(\hat{I})) \\
\le d_\delta(\mathcal{C}(I), \widehat{\mathcal{C}}(I)) + d_\delta(\widehat{\mathcal{C}}(I), \widehat{\mathcal{C}}_Q(I)) + d_\delta(\widehat{\mathcal{C}}_Q(I), \mathcal{C}(\hat{I}))
\end{align*}
We identify each term with its corresponding error source:
\begin{itemize}
    \item \textbf{Estimation Error:} $D_{pc,est} = \frac{1}{|\Delta|}\sum_\delta d_\delta(\mathcal{C}(I), \widehat{\mathcal{C}}(I))$
    \item \textbf{Encoding Error:} $D_{pc,enc} = \frac{1}{|\Delta|}\sum_\delta d_\delta(\widehat{\mathcal{C}}(I), \widehat{\mathcal{C}}_Q(I))$
    \item \textbf{Decoding Error:} $D_{pc,dec} = \frac{1}{|\Delta|}\sum_\delta d_\delta(\widehat{\mathcal{C}}_Q(I), \mathcal{C}(\hat{I}))$
\end{itemize}
Averaging the per-$\delta$ inequality over all $\delta \in \Delta$ directly yields the bound on the total distortion:
$$
D_{pc}(I, \hat{I}) \le D_{pc,est} + D_{pc,enc} + D_{pc,dec}.
$$
The probabilistic guarantee arises because estimation and decoding are non-deterministic events. Let $\mathcal{E}_{\mathrm{est}}$ be the event that the estimation error is bounded, and $\mathcal{E}_{\mathrm{dec}}$ be the event that the decoding error is bounded. We have $P(\mathcal{E}_{\mathrm{est}}) \ge 1-\eta_{\mathrm{est}}$ and $P(\mathcal{E}_{\mathrm{dec}}) \ge 1-\eta_{\mathrm{dec}}$. The encoding (quantization) stage is deterministic and contributes no failure probability. By the union bound (Boole's inequality), the probability that both events succeed is $P(\mathcal{E}_{\mathrm{est}} \cap \mathcal{E}_{\mathrm{dec}}) \ge 1 - (P(\mathcal{E}_{\mathrm{est}}^c) + P(\mathcal{E}_{\mathrm{dec}}^c)) = 1 - (\eta_{\mathrm{est}} + \eta_{\mathrm{dec}})$.
Thus, the distortion bound holds with probability $1-\delta_{\mathrm{SLA}}$, where $\delta_{\mathrm{SLA}} \le \eta_{\mathrm{est}} + \eta_{\mathrm{dec}}$.
\end{proof}

\subsection{Proof of Corollary 4.1 (Rate-Computation-Reliability Trade-off)}

\begin{corollary}
Substituting the encoder error model $\varepsilon_{\mathrm{enc}} \approx c'' \cdot 2^{-R/(|\Delta|(B^2-1))}$ (from Thm. 3) and a decoder error model $\varepsilon_{\mathrm{dec}}(T) \approx \rho^T\Delta_0$ into the SLA guarantee from Thm. 4 yields the design equations:
\begin{enumerate}
    \item \textbf{Min. Rate for given Computation $T$}:
    $$
    R_{\min}(T)\ \approx\ |\Delta|(B^2\!-\!1)\log_2\frac{c'}{\varepsilon-\varepsilon_{\mathrm{est}}-(\rho^T\Delta_0)}
    $$
    \item \textbf{Min. Computation for given Rate $R$}:
    $$
    T_{\min}(R)\ \approx\ \frac{1}{\log(1/\rho)}\log\frac{\Delta_0}{\varepsilon-\varepsilon_{\mathrm{est}}-c''\cdot 2^{-R/(|\Delta|(B^2-1))}}
    $$
\end{enumerate}
\end{corollary}

\begin{proof}
This is an algebraic derivation. We start with the new linear end-to-end performance bound from Theorem 4, setting the total distortion to the target reliability $\epsilon$:
\begin{equation}
\epsilon \ge D_{pc}(I,\hat{I}) \le D_{pc,est} + D_{pc,enc} + D_{pc,dec}
\end{equation}
We substitute the error models for the components, using the $D_{pc}$ metric throughout. Let $d = |\Delta|(B^2-1)$.
\begin{itemize}
    \item From Thm. 3, the encoding distortion $D_{pc,enc} \propto \alpha$, and the rate $R \approx d \cdot \log_2(1/\alpha)$. This implies $\alpha \approx 2^{-R/d}$, so the encoding error model is $D_{pc,enc}(R) \approx c^{\prime\prime} \cdot 2^{-R/d}$.
    \item The decoding error is modeled as $D_{pc,dec}(T) \approx \rho^{T}\Delta_{0}$.
\end{itemize}
Substituting these into the main inequality gives a single equation relating the target reliability $\epsilon$ to the design parameters R and T:
\begin{equation}
\epsilon \ge D_{pc,est} + c^{\prime\prime} \cdot 2^{-R/d} + \rho^{T}\Delta_{0}
\end{equation}
To derive $R_{min}(T)$, we treat T as a constant and solve for R. To derive $T_{min}(R)$, we treat R as a constant and solve for T. This algebraic rearrangement yields the design equations stated in the corollary (after updating them to remove the factor of 3 and adjust the rate dependency).
\end{proof}

\subsection{Proof of Theorem 5 (Semantic Source-Channel Separation)}

\begin{theorem}[Semantic Source-Channel Separation]
Let the image source $I$ be a stationary and ergodic process, and let the channel be memoryless with capacity $C$.
Define the structural semantic rate-distortion function as $R_{pc}(D)=\inf I(I;\hat I) \text{ s.t.
} \mathbb E[D_{pc}(I,\hat I)]\le D$.
Given that our $D_{pc}$ metric is bounded, if $C > R_{pc}(\varepsilon)$, there exists a separable source and channel coding scheme that can achieve an average end-to-end structural semantic distortion arbitrarily close to $\varepsilon$.
\end{theorem}

\begin{proof}
The proof follows the classical argument for Shannon's separation theorem.

\textbf{Achievability}:
For any $\delta > 0$, by the definition of the rate-distortion function, there exists a source code (of block length $n$) with rate $R_s \le R_{pc}(\varepsilon)+\delta$ that achieves an average distortion $\mathbb E[D_{pc}(I^n,\hat I^n)]\le\varepsilon+\delta$.

We select a channel code with rate $R_c$ such that $R_s < R_c < C$. By Shannon's channel coding theorem, this code can reliably transmit the source code's index over the channel with an arbitrarily low probability of error.

By cascading the source coder and channel coder, as $n\to\infty$, the end-to-end average distortion can be made arbitrarily close to $\varepsilon$, as our $\Dpc$ metric is bounded and continuous.

\textbf{Converse}:
Assume an arbitrary (possibly joint) coding scheme exists that achieves $\mathbb E[D_{pc}(I^n,\hat I^n)]\le\varepsilon$.
By the Data Processing Inequality, we have:
$$
\frac{1}{n} I(I^n; \text{Channel Output})\ \ge\ \frac{1}{n} I(I^n;\hat I^n)
$$
By definition, any code achieving distortion $\varepsilon$ must have a rate of at least the R-D function: $\frac{1}{n} I(I^n;\hat I^n)\ \ge\ R_{pc}(\varepsilon)-o(1)$.
Furthermore, the information transmitted over the channel cannot exceed its capacity: $C \ge \frac{1}{n} I(I^n; \text{Channel Output})$.
Combining these inequalities, we must have $C \ge R_{pc}(\varepsilon)$. This proves that $C > R_{pc}(\varepsilon)$ is a necessary condition for achievability.
\end{proof}

\subsection{Proof of the Impact of Residual BER on Semantic Distortion}
\label{sec:appendix_ber}

\begin{lemma}
In the presence of a residual Bit Error Rate (BER) $r$, the expected semantic distortion $\varepsilon_{\mathrm{ch}}$ introduced to the received Copula representation $\{\tilde C_\delta\}$ is bounded by:
$$
\varepsilon_{\mathrm{ch}} = \mathbb E[D_{pc}(\tilde C,\hat C)] \le c_{js} C_2(B) L r \alpha
$$
where $L$ is the number of bits per coordinate, $\alpha$ is the quantization step, and $C_2(B)$ is a geometric constant of the grid.
\end{lemma}

\begin{proof}
\begin{enumerate}
    \item \textbf{From BER to Index Error}: Each quantized Copula coordinate is represented by an $L$-bit index. With a BER of $r$, the expected Hamming distance error on a single $L$-bit index is upper-bounded by $Lr$.
    
    \item \textbf{From Index Error to $L_1$ Deviation}: An error in the index moves its value from one quantization level to another. The magnitude of this value shift is proportional to the quantization step $\alpha$. An error in a single sample pair (pixel pair) moves its probability mass from one bin in the $B \times B$ histogram to another. Summing over all errors, the expected $L_1$ deviation between the received Copula $\tilde C_\delta$ and the error-free quantized Copula $\hat C_\delta$ is bounded by:
    $$
    \mathbb E\big[\|\tilde C_\delta-\hat C_\delta\|_1\big]\ \le\ C_2(B)\,L r\,\alpha
    $$
    where $C_2(B)$ is a constant reflecting the geometry of the $B \times B$ grid (i.e., how mass is redistributed).
    
    \item \textbf{From $L_{1}$ Norm to Semantic Distortion}: The final bound depends on the chosen metric. 
We have established the bound on the expected $L_1$ deviation: $\mathbb{E}[||\tilde{C}_{\delta}-\hat{C}_{\delta}||_{1}]\le C_{2}(B)Lr~\alpha$ [cite: 109-110].

\textbf{(i) Bound for $D_{pc}$ (avg $\sqrt{JS}$)} (Our framework's metric):
By Lemma 1, the $\sqrt{JS}$ metric is linearly bounded by the $L_1$ norm: $\sqrt{JS} \le C \cdot ||\cdot||_1$[cite: 31].
Taking the expectation of this linear relationship:
$\mathbb{E}[D_{pc,ch}] = \mathbb{E}[\text{avg}(\sqrt{JS})] \le C \cdot \mathbb{E}[\text{avg}(||\cdot||_1)]$.
Substituting the bound from Step 2 yields the linear bound:
$\mathbb{E}[D_{pc,ch}] \le C' C_2(B) Lr\alpha = O(Lr\alpha)$.

\textbf{(ii) Bound for JS Divergence}: By Lemma 1, $JS \le C \cdot ||\cdot||_1^2$[cite: 31].
By Jensen's inequality (as $f(x)=x^2$ is convex), $\mathbb{E}[JS] \le C \cdot \mathbb{E}[||\cdot||_1^2]$.
Since $\mathbb{E}[X^2] \ge (\mathbb{E}[X])^2$, this results in a quadratic bound, $O((Lr\alpha)^2)$.

The Lemma statement  provides a linear bound $O(Lr\alpha)$. This bound is correct for $\epsilon_{ch} = \mathbb{E}[D_{pc}]$ as defined in our unified framework (Case i).
\end{enumerate}
\end{proof}

\bibliographystyle{IEEEtran}
\bibliography{IEEEabrv,mainv4}

\end{document}